\documentclass[pdflatex,sn-mathphys-ay]{sn-jnl}

\usepackage{amsmath,amssymb,amsfonts}
\usepackage{mathrsfs}
\usepackage{booktabs}
\usepackage{tabularx}
\usepackage{array}
\usepackage{algorithm}
\usepackage{algorithmicx}
\usepackage{algpseudocode}
\usepackage{listings}
\usepackage{xcolor}
\usepackage{graphicx}
\usepackage{subcaption}
\usepackage{enumitem}
\usepackage{multirow}  % Obbligatorio per \multirow
\usepackage{booktabs}  % Opzionale, per tabelle eleganti
\usepackage{enumitem}

\usepackage{amsmath}

\theoremstyle{thmstyleone}
\newtheorem{theorem}{Theorem}
\newtheorem{proposition}[theorem]{Proposition}

\newtheorem{corollary}[theorem]{Corollary}
\newtheorem{assumption}[theorem]{Assumption}

\theoremstyle{thmstyletwo}

\newtheorem{remark}{Remark}

\theoremstyle{thmstylethree}

\begin{document}

\title{Endogenous Epistemic Weighting under Heterogeneous Information}

\author*[1]{\fnm{Enrico} \sur{Manfredi}}\email{enrico.manfredi@gmail.com}
\affil*[1]{\orgname{Not Affiliated}, \orgaddress{\city{Bologna}, \country{Italy}}}

%https://orcid.org/0000-0001-6125-2816

\abstract{
Collective decision-making requires aggregating multiple noisy 
information channels about an unknown state of the world. 
Classical epistemic justifications of majority rule rely on 
homogeneity assumptions often violated when individual 
competences are heterogeneous. This paper studies endogenous 
epistemic weighting in binary collective decisions. It 
introduces the Epistemic Shared-Choice Mechanism (ESCM), a 
lightweight and auditable procedure that generates bounded, 
issue-specific voting weights from short informational 
assessments. Unlike likelihood-optimal rules, ESCM does not 
require ex ante knowledge of individual competences, but 
infers them endogenously while bounding individual influence. 
Using a central limit approximation under general regularity 
conditions, the paper establishes analytically that bounded 
competence-sensitive monotone weighting strictly increases 
the mean quality of the aggregate signal whenever competence 
is heterogeneous. Numerical comparisons under Beta-distributed 
and segmented mixture competence environments show that these 
mean gains are associated with higher signal-to-noise ratios 
and large-sample accuracy relative to unweighted majority rule.}


%% If you want JEL/MSC fields rendered separately instead of inside the abstract,
%% you can uncomment these:
%% \pacs[JEL Classification]{D71, D72, D83}
%% \pacs[MSC Classification]{91B12, 91B14, 91B06, 91B44}

\maketitle

%==================================================
\section{Introduction}
\label{sec:introduction}
The aggregation of individual judgments into collective decisions is a central
problem in social choice theory and welfare economics. Since Condorcet's seminal
work \citep{condorcet1785}, a long tradition has emphasized the epistemic virtues
of majority rule, showing that collective decisions can track the true state of
the world with increasing accuracy as the electorate grows \citep{grofman1983}.
The Condorcet Jury Theorem (CJT) formalizes this insight under two key
assumptions: conditional independence of individual signals and a common
probability of correctness exceeding one half.

These assumptions are often violated in realistic decision environments.
Individual competences are heterogeneous, information is unevenly distributed,
and many participants rely on weakly informative signals
\citep{downsEconomic1957,converseNature1964,zallerNature1992}. When accuracies
vary substantially, equal weighting conflates informative and uninformative
signals, leading to suboptimal aggregation
\citep{grofman1983,bovensDemocratic2006}. From a statistical perspective, under
classical assumptions of conditionally independent binary signals and known
individual reliabilities in dichotomous choice, likelihood-based aggregation
assigns log-odds weights and maximizes the probability of a correct collective
decision among additive weighting rules
\citep{degrootReaching1974,nitzan1982,shapleyOptimizing1984}. However, these
rules presuppose knowledge of individual reliabilities, which are latent and
typically unobservable.

This raises a fundamental problem: how to approximate competence-sensitive
aggregation when informational quality must be inferred rather than directly
observed. Existing approaches either improve judgment quality prior to
aggregation through deliberation \citep{habermas1996,fishkin2018}, or assign
weights based on exogenous competence indicators \citep{estlund2008,brennan2016}.
In both cases, the aggregation rule itself does not endogenously incorporate
inferred informational reliability. Recent work on endogenous information
acquisition \citep{persicoCommittee2004,martinelliWould2006,bhattacharyaVoting2017}
studies how institutions shape incentives to acquire signals, but typically
treats epistemic quality as exogenous at the aggregation stage.

This paper studies an alternative approach: \emph{endogenous epistemic weighting}.
It introduces the Epistemic Shared-Choice Mechanism (ESCM), a class of
procedures that infer issue-specific signal reliability through short
assessments and map it into bounded voting weights. The paper does not claim to
establish general Nitzan--Paroush optimality under realistic observability
constraints. Rather, it studies a procedural attempt to approximate
competence-sensitive aggregation when reliabilities are latent and individual
influence must remain bounded.

The analysis is conducted on a stylized, reliability-based representation of
endogenous weighting, and the formal results pertain to this representation.
The peer-generated assessment mechanism is proposed as a procedural
implementation of this idea, but its full psychometric and strategic properties
remain open for future work.

The contribution proceeds in four steps. First, the paper introduces a formal
framework for endogenous bounded weighting in collective choice and relates the
log-odds version of ESCM to a likelihood-based benchmark under a consistency
assumption on the recovery of latent reliabilities. Second, it derives a
Gaussian approximation for weighted aggregation under general regularity
conditions, showing that, within this CLT framework, large-sample performance
is approximately determined by the induced signal-to-noise ratio
$\mu_T/\sigma_T$ as defined in Section~\ref{sec:clt-general}.  Third, it
illustrates numerically that, for the specific weighting rules and competence
environments studied here, this mean improvement is associated with higher
signal-to-noise ratios and large-sample accuracy. Fourth, it
establishes a general mean-improvement result under heterogeneous competence
for bounded monotone weighting under a suitable normalization.

Under heterogeneous competence distributions with non-zero variance, any bounded,
non-decreasing epistemic weighting function that is non-constant on the support
of the competence distribution and normalized so that $E_f[w(p)] = 1$ strictly
increases the mean signed aggregate signal relative to unweighted majority rule. 
The paper does not establish universal
signal-to-noise-ratio dominance for all such weighting rules. Instead, gains in
signal-to-noise ratio and collective accuracy are shown numerically for the
specific weighting specifications and competence distributions examined in
Sections~\ref{sec:unimodal} and~\ref{sec:cmm}.

The remainder of the paper is organized as follows.
Section~\ref{sec:model} introduces the formal framework and the
likelihood-based benchmark. Section~\ref{sec:escm} describes the ESCM
procedure. Section~\ref{sec:struct} discusses the structural properties of
ESCM. Section~\ref{sec:clt-general} derives the Gaussian approximation for
weighted aggregation. Section~\ref{sec:unimodal} evaluates ESCM under
Beta-distributed competence. Section~\ref{sec:cmm} examines segmented
informational structures. Section~\ref{sec:dominance} establishes the general
mean-improvement result under heterogeneous competence.
Section~\ref{sec:conclusion} concludes.

%%%%%%%%%%%%%%%%%%%%%%%%%%%%%%%%%%%%%%%%%%%%%%%%%%%%%%%%%%%%
\section{Model of the Decision}
\label{sec:model}

\subsection{Binary Epistemic Collective Choice}

We consider a binary collective decision problem with unknown true state
$x^{\ast} \in X = \{0,1\}$. The objective of the group is to select the
alternative that matches the true state.

There is a finite set of participants $N = \{1, \ldots, n\}$. Each participant
$i$ observes a private signal $\xi_i \in \{0,1\}$ about $x^{\ast}$.
Conditional on $x^{\ast}$, signals are assumed to be independent across
individuals, so that they can be interpreted as parallel information channels
in the classical epistemic sense \citep{condorcet1785,grofman1983}.

Participants cast votes $v_i \in \{0,1\}$, and we assume sincere voting, so
that
\[
v_i = \xi_i.
\]
This allows the analysis to isolate informational aggregation from strategic
behavior.

Each participant $i$ is characterized by a competence parameter
$p_i \in [0,1]$, defined as the probability that the participant's signal
matches the true state:
\[
p_i = \Pr(\xi_i = x^{\ast}).
\]
Competences may differ arbitrarily across individuals, reflecting
heterogeneous and issue-specific information. No assumption is imposed here
that all participants are better than random, although collective performance
depends on the distribution of the $\{p_i\}$.

For later use, it is convenient to introduce the correctness-based signed
signal
\[
Y_i = 2\,\mathbf{1}\{\xi_i = x^{\ast}\} - 1 \in \{-1,+1\},
\]
so that
\[
\Pr(Y_i = +1 \mid p_i) = p_i,
\qquad
\Pr(Y_i = -1 \mid p_i) = 1-p_i,
\]
and therefore
\[
\mathbb{E}[Y_i \mid p_i] = 2p_i - 1.
\]
This correctness-based signed representation differs from the vote coding
$2v_i-1 \in \{-1,+1\}$ used below to define aggregation in favor of
alternative $1$. The former is convenient for epistemic performance analysis,
while the latter is convenient for describing the decision rule.

\subsection{Aggregation Rules and the Equal-Weight Benchmark}

An aggregation rule maps the vector of individual votes
$\mathbf{v} = (v_1,\dots,v_n) \in \{0,1\}^n$ into a collective decision
$\hat{x} \in \{0,1\}$.

The main equal-weight benchmark considered in the paper is unweighted majority
rule, which assigns equal influence to all participants. In binary form, it
selects alternative $1$ whenever
\[
\sum_{i=1}^n v_i \;\ge\; \frac{n}{2},
\]
and alternative $0$ otherwise, with an arbitrary tie-breaking rule if $n$ is
even.

Equivalently, define the aggregate score
\[
T_n : \{0,1\}^n \to \mathbb{R}, \qquad
T_n(\mathbf{v}) \;=\; \sum_{i=1}^n v_i,
\]
and select alternative $1$ whenever $T_n(\mathbf{v}) \ge \frac{n}{2}$.

Under homogeneous competence and the usual Condorcet assumptions, this rule
reduces to the classical equal-weight benchmark; under heterogeneous
competence, however, it does not exploit differences in informational quality
across participants.

More generally, a weighted aggregation rule assigns a weight $w_i$ to each
participant and selects alternative $1$ whenever
\[
T_n^w(\mathbf{v}) \;=\; \sum_{i=1}^n w_i\, v_i \;\ge\; \frac{1}{2}\sum_{i=1}^n w_i,
\]
and alternative $0$ otherwise.

The central question of the paper is how such weights can be generated
endogenously, using observable assessment performance as a proxy for latent
reliability while keeping individual influence bounded.

\begin{remark}[Multi-option extension]
The analysis in this paper focuses on binary collective decisions. For
$|X|>2$, the informational environment must in general be described by a full
confusion structure
\[
p_i(x,x') = \Pr(\xi_i = x' \mid x^{\ast}=x),
\qquad x,x' \in X,
\]
rather than by a single competence parameter $p_i$. Weighted plurality remains
well defined in that setting, but the likelihood-based benchmark generalizes to
a multi-class problem. This extension is left for future work.
\end{remark}

%%%%%%%%%%%%%%%%%%%%%%%%%%%%%%%%%%%%%%%%%%%%%%%%%%%%%%%%%%%%
\section{An Epistemic Shared--Choice Mechanism for Collective Decisions}
\label{sec:escm}

Building on the binary benchmark above, we introduce a procedural architecture
for endogenous epistemic weighting. The \emph{Epistemic Shared--Choice
Mechanism} (ESCM) is not presented here as a fully validated institution.
Rather, it is a structured procedure that uses peer-generated assessment items
to construct observable proxies for issue-specific reliability, maps those
proxies into bounded aggregation coefficients, and then aggregates votes
accordingly.

Throughout this section, \emph{assessment items} are knowledge-testing items
used to evaluate participants' informational performance on the issue under
consideration. They are distinct from the \emph{alternatives}
$x \in X = \{0,1\}$, which are the objects of the collective decision at the
aggregation stage.

\paragraph{Step 1: Item authoring}
Let $N = \{1, \ldots, n\}$ be the set of participants. Each participant proposes
$l_w$ multiple-choice items on the relevant topic. Let $Q_0$ denote the initial
item pool, with
\[
|Q_0| = n\,l_w.
\]

\paragraph{Step 2: Peer review}
Each item $j \in Q_0$ is assigned to $m$ reviewers, with self-review excluded.
In a balanced design, each participant therefore reviews
\[
l_r = m\,l_w
\]
items.

For each assigned item, reviewer $i$ provides ratings in $[0,1]$ along the
following dimensions:
\begin{itemize}
  \item \textbf{relevance} $g_{ij}$: direct pertinence to the decision scope;
  \item \textbf{clarity} $c_{ij}$: clear, unambiguous, and easily understandable phrasing;
  \item \textbf{absence of bias} $b_{ij}$: neutral wording, no leading or loaded framing;
  \item \textbf{factual correctness} $f_{ij}$: factual statements are judged to be objectively true and verifiable;
  \item \textbf{scientific accuracy} $a_{ij}$: consistency with established scientific knowledge;
  \item \textbf{principle adherence} $e_{ij}$: compliance with relevant norms and standards;
  \item \textbf{difficulty level} $d_{ij}$: level of knowledge required for accurate response.
\end{itemize}

Items whose average quality evaluation falls below a prescribed threshold $\tau$ are
discarded, producing a validated pool $Q \subseteq Q_0$. The difficulty ratings
$d_{ij}$ are retained for use in the balancing step below.

This review stage is intended to screen for basic quality, perceived validity,
and difficulty; it does not by itself guarantee factual truth or full
scientific reliability.

\begin{remark}[Feasibility constraints]
The review design imposes the accounting identity $l_r = m\,l_w$ in any
balanced implementation. After filtering, Step~3 requires enough surviving
items so that each participant can be assigned $l_a$ items that they neither
authored nor reviewed. 
\end{remark}

\paragraph{Step 3: Pool construction}
Each participant $i \in N$ receives $l_a$ items from $Q$ at random, excluding
both items proposed and items previously reviewed by $i$. Let $A(i)$ denote the
set of assessment items assigned to participant $i$.

To improve comparability across participants, questionnaires are balanced as
far as feasible by estimated item difficulty. For each item $j \in Q$, let
\[
\hat d_j = \frac{1}{m}\sum_{i \in R(j)} d_{ij},
\]
where $R(j)$ denotes the set of reviewers assigned to item $j$ in Step~2. The
assignment procedure aims to keep the distribution of $\hat d_j$ over the items
in $A(i)$ approximately similar across participants.

\paragraph{Step 4: Individual assessment}
Let $Y_{ij} \in \{0,1\}$ indicate whether participant $i$ answered item
$j \in A(i)$ correctly. Each item has $q$ response options. To discourage
random guessing while preserving universal participation, ESCM uses the
truncated linear score
\[
s_i
=
\max\!\left\{
s_{\min},
\sum_{j \in A(i)}
\left(
\mathbf{1}\{Y_{ij}=1\}
-
\frac{1}{q-1}\mathbf{1}\{Y_{ij}=0\}
\right)
\right\},
\qquad s_{\min} > 0.
\]
Under random guessing, the pre-truncation score has expectation zero. The floor
$s_{\min}$ does not imply that a random responder always receives the minimal
realized score; rather, it guarantees a strictly positive lower bound on the
realized score and therefore prevents degenerate zero values in the subsequent
normalization step.

After completing the assessment, each participant submits a vote
$v_i \in \{0,1\}$ on the binary collective decision.

\paragraph{Step 5: Construction of aggregation coefficients}
Assessment scores are normalized as
\[
\bar{s}_i = \frac{s_i}{l_a}
\in
\left[\frac{s_{\min}}{l_a},\,1\right].
\]
These normalized scores are treated as observable proxies for latent
issue-specific reliability and mapped into aggregation coefficients through a
monotone function
\[
g:[0,1]\to\mathbb{R},
\qquad
w_i = g(\bar{s}_i).
\]

Three specifications are considered in this paper.

The \emph{linear} specification
\[
g(\bar{s}_i)=\bar{s}_i
\]
preserves ordinal ranking and provides a transparent benchmark.

The \emph{power} specification
\[
g_k(\bar{s}_i)=\bar{s}_i^{\,k},
\qquad k>0,
\]
allows the designer to tune epistemic selectivity: $k>1$ amplifies differences
at the top of the score distribution, while $k\in(0,1)$ compresses them.

The \emph{regularized log-odds} specification
\[
g_{\varepsilon}(\bar{s}_i)
=
\log\!\left(
\frac{\bar{s}_i+\varepsilon}{1-\bar{s}_i+\varepsilon}
\right),
\qquad \varepsilon>0,
\]
provides a bounded likelihood-oriented transformation and yields coefficients in
the interval $[\omega_{\min},\omega_{\max}]$, where
\[
\omega_{\min}
=
\log\!\left(
\frac{\tfrac{s_{\min}}{l_a}+\varepsilon}
     {1-\tfrac{s_{\min}}{l_a}+\varepsilon}
\right),
\qquad
\omega_{\max}
=
\log\!\left(\frac{1+\varepsilon}{\varepsilon}\right).
\]

\begin{remark}[Negative coefficients under regularized log-odds]
If $\bar{s}_i < 1/2$, then $g_{\varepsilon}(\bar{s}_i)$ may be negative.
Under this specification, the mechanism is therefore better interpreted as
\emph{signed evidence accumulation} rather than as non-negative weighted voting:
a participant with a sufficiently low assessment score contributes evidence
against the option they support. Designers wishing to preserve non-negative
influence for all participants should use the linear or power specifications.
\end{remark}

\paragraph{Step 6: Binary aggregation}
Since the analysis in this paper focuses on binary collective decisions, the
collective statistic induced by ESCM is
\[
T_n = \sum_{i=1}^n w_i v_i.
\]
The mechanism selects alternative $1$ whenever
\[
T_n \;\ge\; \frac{1}{2}\sum_{i=1}^n w_i,
\]
and alternative $0$ otherwise.

When all $w_i \ge 0$, this rule is weighted majority. Under the regularized
log-odds specification, however, the same formula may involve signed
coefficients and is more accurately interpreted as a binary evidence-aggregation
rule.

\begin{remark}[Strategic robustness]
The present analysis abstracts from strategic behavior at the assessment stage.
Potential manipulation of item authorship, peer-review scores, item selection,
or assessment responses is outside the scope of this paper. The formal results
below concern the aggregation properties of the induced coefficients, not the
incentive-compatibility of the full institutional procedure.
\end{remark}

%==================================
\section{Structural Properties of ESCM}
\label{sec:struct}

The following remarks and results clarify how ESCM should be interpreted
analytically. The section does not attempt to validate the full institutional
procedure in all its psychometric or strategic dimensions. Rather, it studies
stylized properties of the score-to-coefficient mapping induced by the
mechanism and shows how those properties depend on assessment length and design
parameters.

\begin{remark}[Noise as weakly discriminating information]
Low-information participants need not generate manifestly false assessment
items. They may instead generate items that are weakly discriminating: such
items may appear locally plausible, yet fail to separate more informed from
less informed respondents in a reliable way.

Within ESCM, this kind of epistemic noise is not assumed away \emph{ex ante}.
Instead, it is only partially filtered through the assessment stage: items that
induce little variation in participant performance tend to contribute less to
score dispersion and therefore less to differentiation in the induced
aggregation coefficients.
\end{remark}

\begin{remark}[Parameter flexibility]
ESCM can be adapted to different decision environments by varying the
assessment design, the review procedure, and the score-to-coefficient map. The
parameters $(q,l_w,l_r,l_a,m,s_{\min},\tau,k,\varepsilon)$ govern the balance
between assessment precision, participant burden, epistemic selectivity, and
the dispersion of influence.

Among these parameters, $l_a$ primarily controls the precision of score
estimation relative to assessment cost. The parameter $k$ controls the
selectivity of the power transformation, while $\varepsilon$ smooths the
regularized log-odds transformation near the boundaries.
\end{remark}

\begin{remark}[Assessment noise as an idealized approximation]
\label{lem:noise}
Let $C_i \mid p_i \sim \mathrm{Binomial}(l_a,p_i)$ denote an idealized
benchmark count of correct answers for participant $i$, with $p_i \in (0,1)$,
and let $\bar C_i = C_i/l_a$. This benchmark does not coincide exactly with
the normalized ESCM score $\bar s_i=s_i/l_a$ from
Section~\ref{sec:escm}. In particular, it does not model peer-generated item
selection, review filtering, or truncation at $s_{\min}$. It does, however,
capture one leading source of estimation noise in analytically transparent
form: finite assessment length.

The order-of-magnitude bounds below should be read as holding uniformly for
$p_i$ in compact subsets of $(0,1)$ and as heuristic approximations rather
than exact finite-sample inequalities.

Consider the power transformation $g_k(\bar C_i)=\bar C_i^{\,k}$ for $k>0$.
Then:
\begin{enumerate}[label=(\roman*)]
\item \textbf{Bias bound.}
A second-order Taylor expansion implies
\[
\bigl|\mathbb{E}[g_k(\bar C_i)\mid p_i]-p_i^{\,k}\bigr|
=
O\!\left(\frac{k^2}{l_a}\right),
\]
so that, for fixed $k$, the distortion induced by binomial assessment noise
vanishes as $l_a\to\infty$.

\item \textbf{Selectivity.}
The steepness ratio between a perfect score and a single mistake satisfies
\[
\frac{g_k\!\left(\dfrac{l_a-1}{l_a}\right)}{g_k(1)}
=
\left(1-\frac{1}{l_a}\right)^{k}
\approx
\exp\!\left(-\frac{k}{l_a}\right),
\]
showing that the operative design parameter is the ratio $k/l_a$: larger
values increase epistemic selectivity at the cost of greater sensitivity to
assessment noise.

\item \textbf{Sample complexity.}
For a target approximation error $\delta>0$, the preceding bound suggests the
order-of-magnitude requirement
\[
l_a \gtrsim \frac{k^2}{\delta},
\]
providing a heuristic lower bound on the number of assessment items required
for stable estimation of transformed scores.
\end{enumerate}
\end{remark}

\begin{proposition}[Asymptotic relation to reliability-based score transforms]
\label{prop:asymptotic}
Assume that $\bar s_i \xrightarrow{P} p_i$ as $l_a\to\infty$, with
$p_i\in(0,1)$ for all $i\in N$.

\begin{enumerate}[label=(\roman*)]
\item For any $k>0$,
\[
g_k(\bar s_i)=\bar s_i^{\,k}\xrightarrow{P} p_i^{\,k}.
\]

\item For any fixed $\varepsilon>0$,
\[
g_{\varepsilon}(\bar s_i)
=
\log\!\left(\frac{\bar s_i+\varepsilon}{1-\bar s_i+\varepsilon}\right)
\xrightarrow{P}
\log\!\left(\frac{p_i+\varepsilon}{1-p_i+\varepsilon}\right).
\]

\item For any $\delta\in(0,1/2)$,
\[
\sup_{p\in[\delta,\,1-\delta]}
\left|
\log\!\left(\frac{p+\varepsilon}{1-p+\varepsilon}\right)
-
\log\!\left(\frac{p}{1-p}\right)
\right|
\longrightarrow 0
\qquad\text{as }\varepsilon\downarrow 0.
\]
\end{enumerate}

Hence, under consistent score recovery, the ESCM transformations converge to
their corresponding reliability-based counterparts; moreover, on competence
intervals bounded away from $0$ and $1$, the regularized log-odds
transformation provides a bounded approximation to the unregularized
likelihood-oriented benchmark when $\varepsilon$ is small.
\end{proposition}

\begin{proof}
Parts (i) and (ii) follow directly from the continuous mapping theorem, since
$x\mapsto x^k$ and
$x\mapsto \log((x+\varepsilon)/(1-x+\varepsilon))$
are continuous on $(0,1)$ for fixed $\varepsilon>0$.

For part (iii), define
\[
h_{\varepsilon}(p)
=
\log\!\left(\frac{p+\varepsilon}{1-p+\varepsilon}\right)
-
\log\!\left(\frac{p}{1-p}\right).
\]
For each fixed $p\in(0,1)$, one has $h_{\varepsilon}(p)\to 0$ as
$\varepsilon\downarrow 0$. Since $h_{\varepsilon}$ is continuous in $p$ and the
interval $[\delta,1-\delta]$ is compact and bounded away from the singular
points $0$ and $1$, the convergence is uniform on that interval.
\end{proof}

This proposition concerns the assessment-precision regime $l_a\to\infty$ for a
fixed participant-level score recovery problem. The large-population regime
$n\to\infty$ studied in the next section is analytically distinct and treats
the induced coefficient map in reduced form.

\begin{proposition}[Monotonicity, boundedness, and sign structure]
\label{prop:weights}
Fix design parameters with $0<s_{\min}\le l_a$, and let
\[
\bar s_i \in \left[\frac{s_{\min}}{l_a},\,1\right].
\]
Then the coefficient maps used in ESCM satisfy the following properties:

\begin{enumerate}[label=(\roman*)]
\item The linear map $g(\bar s)=\bar s$, the power map
$g_k(\bar s)=\bar s^{\,k}$ for $k>0$, and the regularized log-odds map
\[
g_{\varepsilon}(\bar s)
=
\log\!\left(\frac{\bar s+\varepsilon}{1-\bar s+\varepsilon}\right)
\]
are measurable and non-decreasing on
$\left[\frac{s_{\min}}{l_a},1\right]$.

\item The linear and power maps are non-negative and bounded on this interval.

\item The regularized log-odds map is bounded on this interval and satisfies
\[
g_{\varepsilon}(\bar s) < 0 \quad \text{if } \bar s < \frac12,
\qquad
g_{\varepsilon}(\bar s) = 0 \quad \text{if } \bar s = \frac12,
\qquad
g_{\varepsilon}(\bar s) > 0 \quad \text{if } \bar s > \frac12.
\]
\end{enumerate}

Consequently, the linear and power specifications induce non-negative
weighted-majority rules, whereas the regularized log-odds specification
induces a signed evidence-aggregation rule.
\end{proposition}

\begin{proof}
All claims follow directly from the elementary properties of the three
functions on the interval
$\left[\frac{s_{\min}}{l_a},1\right]$.
\end{proof}

Propositions~\ref{prop:asymptotic} and~\ref{prop:weights} play different roles.
Proposition~\ref{prop:asymptotic} relates ESCM score transformations to
stylized reliability-based benchmarks under consistent score recovery.
Proposition~\ref{prop:weights} establishes the boundedness and monotonicity
properties needed for the large-sample analysis developed in the next section.

\begin{remark}[Epistemic accuracy and inequality of influence]
\label{rmk:accuracy-inequality}
By construction, ESCM converts informational heterogeneity into heterogeneity
of aggregation coefficients. This does not by itself normatively justify
unequal influence. What the mechanism does provide is a way to make that
trade-off explicit and measurable.

For non-negative specifications such as the linear and power maps, standard
concentration measures such as the Herfindahl index
\[
H=\sum_i w_i^2
\]
and the Gini coefficient applied to $\{w_i\}$ can be used to summarize the
dispersion of influence. Under signed specifications such as regularized
log-odds, analogous summaries are better interpreted as measures of coefficient
dispersion rather than influence concentration, or else applied to suitably
normalized absolute coefficients.
\end{remark}

\begin{remark}[Computational tractability]
For fixed design parameters, the main stages of ESCM are computationally
tractable in population size under standard assignment implementations. Item
review and assessment scale linearly in the number of assigned reviews and
responses, while binary aggregation is linear in $n$.

The main practical trade-off is therefore not computational feasibility but the
administrative cost of increasing assessment precision through larger values of
$l_a$, $l_r$, and $m$. More exact balancing requirements in questionnaire
construction may require heuristic or approximate assignment procedures.
\end{remark}

Together, these properties show that ESCM is best interpreted not as a single
fixed voting rule, but as a family of procedurally generated aggregation
schemes with bounded and monotone score-to-coefficient maps whose epistemic
behavior depends on assessment precision and on the chosen score
transformation.

%=============================
\section{Central Limit Approximation under General Information Distributions}
\label{sec:clt-general}

To evaluate the large-sample epistemic performance of ESCM relative to the
equal-weight benchmark, we use a Gaussian approximation for the aggregate
signal based on the Lindeberg--Feller central limit theorem
\citep{lindebergNeue1922,fellerKolmogorovSmirnov1948}. The approximation is
formulated in reduced form: rather than modeling the full assessment procedure,
it studies the aggregate statistic induced by a bounded competence-dependent
coefficient map.

\subsection{Signal Model and Regularity Conditions}

Consider the binary epistemic setting introduced in
Section~\ref{sec:model}. Let $Y_i \in \{-1,+1\}$ denote the correctness-based
signed signal, so that
\[
\Pr(Y_i = +1 \mid p_i) = p_i,
\qquad
\Pr(Y_i = -1 \mid p_i) = 1 - p_i,
\]
where $p_i \in [0,1]$ denotes individual competence. Conditional on $p_i$,
signals are independent across individuals and satisfy
\[
\mathbb{E}[Y_i \mid p_i] = 2p_i - 1.
\]

In the present section, ESCM is represented in reduced form through a bounded
coefficient map $w:[0,1]\to\mathbb{R}$, where $w(p_i)$ denotes the aggregation
coefficient associated with competence level $p_i$. The resulting aggregate
signal is
\[
T_n = \sum_{i=1}^n w(p_i)\,Y_i.
\]

\begin{assumption}[Regularity conditions]
\label{ass:regularity}
The competences $p_i$ are independently drawn from a distribution $f$ on
$[0,1]$, and the coefficient map $w:[0,1]\to\mathbb{R}$ is measurable and
bounded. In addition, the induced variance
\[
\sigma_T^2
=
\mathbb{E}_f\!\left[w(p)^2\bigl(1-(2p-1)^2\bigr)\right]
+
\mathrm{Var}_f\!\left(w(p)(2p-1)\right)
\]
is strictly positive.
\end{assumption}

Boundedness of $w$ is satisfied by construction under ESCM for the score
transformations considered in Section~\ref{sec:struct}. In the specifications
studied below, $w$ is also taken to be non-decreasing in competence, reflecting
the epistemic requirement that more competent participants receive weakly
larger aggregation coefficients, although monotonicity is not needed for the
CLT itself.

\subsection{CLT Approximation for Weighted Aggregation}

\begin{proposition}[Gaussian approximation of the collective signal]
\label{prop:clt}
Under Assumption~\ref{ass:regularity}, define
\[
\mu_T = \mathbb{E}_f\!\left[w(p)(2p-1)\right],
\]
and
\[
\sigma_T^2
=
\mathbb{E}_f\!\left[w(p)^2\bigl(1-(2p-1)^2\bigr)\right]
+
\mathrm{Var}_f\!\left(w(p)(2p-1)\right).
\]
Then:
\begin{enumerate}[label=(\roman*)]
\item \textbf{Moment scaling.}
The mean and variance of $T_n$ satisfy
\[
\mathbb{E}[T_n] = n\,\mu_T,
\qquad
\mathrm{Var}(T_n) = n\,\sigma_T^2.
\]

\item \textbf{Asymptotic normality.}
As $n \to \infty$,
\[
\frac{T_n - n\mu_T}{\sqrt{n\,\sigma_T^2}}
\xrightarrow{d}
\mathcal{N}(0,1).
\]
\end{enumerate}
\end{proposition}

\begin{proof}
Let
\[
X_i := w(p_i)Y_i.
\]
Under Assumption~\ref{ass:regularity}, the variables $\{X_i\}_{i=1}^n$ are
independent and identically distributed. Moreover,
\[
\mathbb{E}[X_i]
=
\mathbb{E}_f\!\left[\mathbb{E}[w(p_i)Y_i \mid p_i]\right]
=
\mathbb{E}_f\!\left[w(p_i)(2p_i-1)\right]
=
\mu_T.
\]
Similarly, by the law of total variance,
\[
\mathrm{Var}(X_i)
=
\mathbb{E}_f\!\left[\mathrm{Var}(w(p_i)Y_i \mid p_i)\right]
+
\mathrm{Var}_f\!\left(\mathbb{E}[w(p_i)Y_i \mid p_i]\right)
=
\sigma_T^2.
\]
This proves part (i). Since $w(\cdot)$ is bounded and $Y_i\in\{-1,+1\}$, the
summands $X_i$ are uniformly bounded and therefore have finite second moments.
Hence the Lindeberg condition is satisfied, and the Lindeberg--Feller CLT
applies, yielding part (ii).
\end{proof}

\subsection{Distribution-Robustness of CJT--ESCM Comparisons}

\begin{corollary}[Moment-based robustness of the Gaussian comparison]
\label{cor:robustness}
Under Assumption~\ref{ass:regularity}, for large $n$ the probability that the
aggregate signal has the correct sign (i.e., favors the true state) is
approximated by
\[
\Pr(T_n > 0)
\;\approx\;
\Phi\!\left(\sqrt{n}\,\frac{\mu_T}{\sigma_T}\right),
\]
where $\Phi$ denotes the standard normal CDF. Hence, within this Gaussian
approximation, large-sample epistemic performance depends on the competence
distribution $f$ only through the induced moments $\mu_T$ and $\sigma_T^2$.
Replacing a specific parametric family for $f$ with a general distribution does
not alter the functional form of the approximation; it changes only the
numerical values of these moments.
\end{corollary}

\begin{proof}
By Proposition~\ref{prop:clt}, the standardized statistic converges in
distribution to a standard normal random variable. Rewriting $\Pr(T_n>0)$ in
standardized form yields the stated approximation.
\end{proof}

The role of ESCM in this reduced-form analysis is therefore to alter the
effective moments $(\mu_T,\sigma_T^2)$ of the aggregate signal through
endogenous coefficient assignment, rather than to rely on any particular
parametric specification of the competence distribution.

%========================================
%=============================
%=============================
\section{Unimodal Competence: Beta-Distributed Heterogeneity}
\label{sec:unimodal}

\subsection{The Beta Benchmark}
\label{subsec:beta}

To study epistemic aggregation under heterogeneous but unimodal information, we
begin with the canonical case in which individual competences follow a Beta
distribution. The Beta family provides a flexible benchmark for bounded
heterogeneity on $(0,1)$ and has been widely used in the epistemic social
choice literature
\citep{grofman1983,bovensDemocratic2006,dietrichDeliberation2025}.

Let $p_i \in (0,1)$ and assume
\[
p_i \sim \mathrm{Beta}(\alpha,\beta),
\qquad
\alpha,\beta>1,
\]
so that competences are concentrated around a single interior mode. The mean
and variance are
\[
\mu = \mathbb{E}[p_i] = \frac{\alpha}{\alpha+\beta},
\qquad
\sigma^2 = \mathrm{Var}(p_i)
= \frac{\alpha\beta}{(\alpha+\beta)^2(\alpha+\beta+1)}.
\]
Symmetric cases with $\alpha=\beta$ describe electorates concentrated around a
common intermediate competence level, whereas asymmetric cases allow for
skewed populations with relatively more informed or less informed electorates.

In the reduced-form framework of Section~\ref{sec:clt-general}, the
equal-weight benchmark depends on the competence distribution only through the
induced moments of the aggregate signal. The Beta family is useful because it
allows the location and dispersion of competence to vary in a controlled and
interpretable way while maintaining a unimodal benchmark structure.

\subsection{Parameterization and Numerical Setup}
\label{subsec:beta-method}

The analysis is conducted over a grid of Beta distributions parameterized by
their mean $\mu$ and standard deviation $\sigma$. Since Beta moments satisfy
\[
\sigma^2
=
\frac{\mu(1-\mu)}{\alpha+\beta+1},
\]
feasible parameter pairs must lie in the Beta-admissible region
\[
\sigma^2 < \mu(1-\mu).
\]
This region defines the set of unimodal competence environments considered in
the figures below.

All results fix $n=501$. For the ESCM specifications considered below, the
assessment length is set to $l_a=10$, and this parameter enters through the
assessment-induced mapping from competence to aggregation coefficients rather
than through an explicit simulation of the full institutional procedure. For
each specification of the competence distribution and coefficient map, the
reduced-form moments $\mu_S$ and $\sigma_S^2$ of
Section~\ref{sec:clt-general} are computed by numerical integration over the
competence distribution. For the regularized log-odds specification, the
regularization parameter is set to $\varepsilon=0.01$. The approximate success
probability is then evaluated as
\[
\Pr(S_n>0)\approx \Phi\!\left(\sqrt{n}\,\frac{\mu_S}{\sigma_S}\right),
\]
so that the figures below should be interpreted as diagnostics for the
large-sample Gaussian approximation rather than as finite-sample guarantees for
an implemented ESCM.

\subsection{Results under Beta-Distributed Competence}
\label{subsec:beta-results}

\paragraph{Equal-weight benchmark.}
Figure~\ref{fig:beta_cjt} reports the approximate success probability under the
equal-weight rule. Within the Gaussian approximation and for the Beta
environments considered here, performance is driven mainly by the mean
competence level $\mu$, with a sharp transition near $\mu=0.5$. In the
classical homogeneous CJT benchmark, dispersion plays no independent role
once the mean is fixed; in the present heterogeneous Beta setting, the
numerical results indicate that variation in $\sigma$ has only a limited
additional effect on the equal-weight rule.

\begin{figure}[t]
\centering
\includegraphics[width=0.48\linewidth]{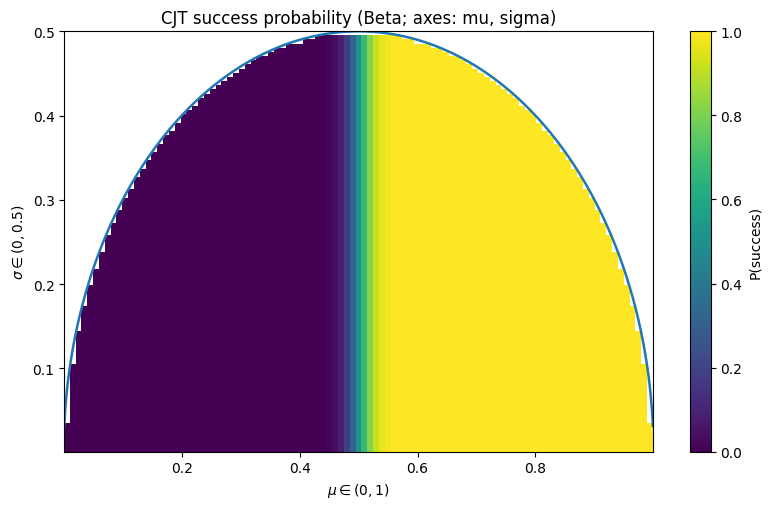}
\caption{Approximate equal-weight success probability under Beta-distributed competence.}
\label{fig:beta_cjt}
\end{figure}

\paragraph{ESCM with linear coefficients.}
Figure~\ref{fig:beta_linear} reports ESCM performance under the linear
specification
\[
g(\bar{s})=\bar{s}.
\]
Unlike the equal-weight benchmark, the induced approximate success probability
depends on both $\mu$ and $\sigma$. Gains are concentrated in regions where
average competence lies near the indifference threshold and competence
dispersion is non-negligible. In those environments, competence-sensitive
reweighting improves the contribution of better-informed participants without
requiring extreme separation in the competence distribution. As $\mu$ moves
further away from $0.5$, the room for improvement narrows because the
equal-weight rule already achieves high approximate accuracy.

\begin{figure}[t]
\centering
\begin{subfigure}[t]{0.48\linewidth}
  \centering
  \includegraphics[width=\linewidth]{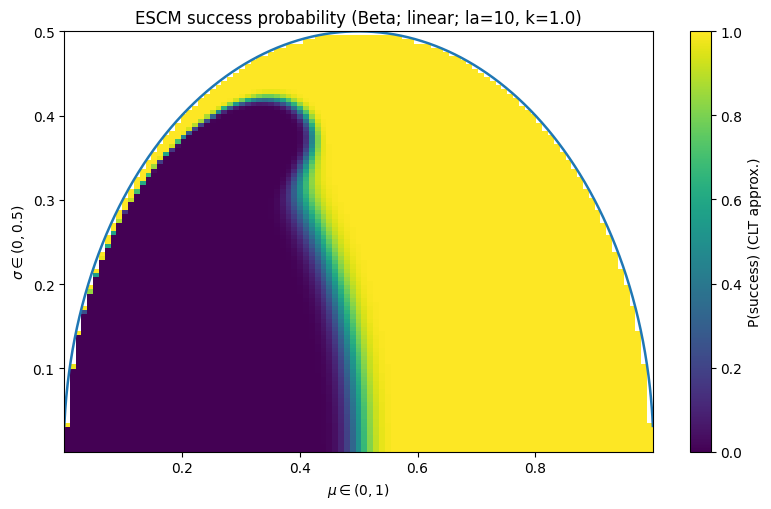}
  \caption{Approximate ESCM success probability}
\end{subfigure}\hfill
\begin{subfigure}[t]{0.48\linewidth}
  \centering
  \includegraphics[width=\linewidth]{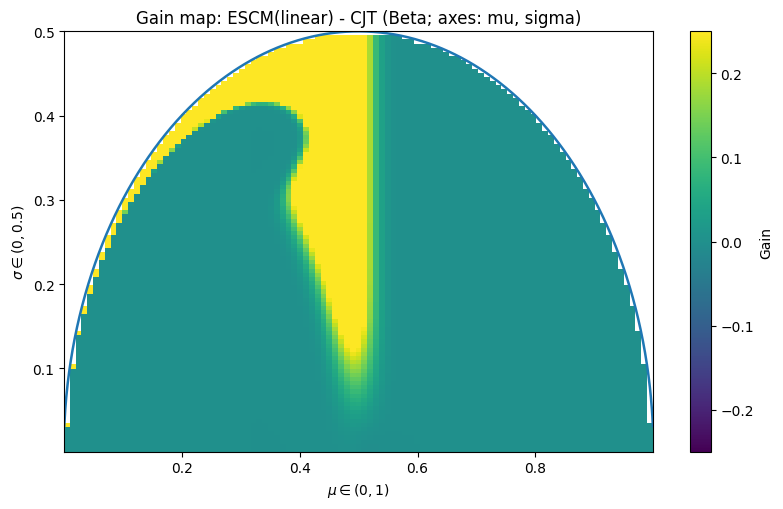}
  \caption{Gain relative to the equal-weight benchmark}
\end{subfigure}
\caption{ESCM with linear coefficients under Beta-distributed competence ($l_a=10$, $k=1$).}
\label{fig:beta_linear}
\end{figure}

\paragraph{ESCM with regularized log-odds coefficients.}
Figure~\ref{fig:beta_logit} reports ESCM performance under the regularized
log-odds specification. Relative to the linear specification, gains are larger
and extend over a broader region of the $(\mu,\sigma)$ plane. This pattern is
consistent with a more selective competence-sensitive transformation of
assessment scores in this numerical environment. Since the regularized log-odds
map may generate negative coefficients for low scores, the mechanism is best
interpreted here as a form of signed evidence aggregation rather than as a
non-negative weighted-majority rule. The regularization parameter
$\varepsilon=0.01$ keeps the induced coefficients bounded near the score
boundaries.

\begin{figure}[t]
\centering
\begin{subfigure}[t]{0.48\linewidth}
  \centering
  \includegraphics[width=\linewidth]{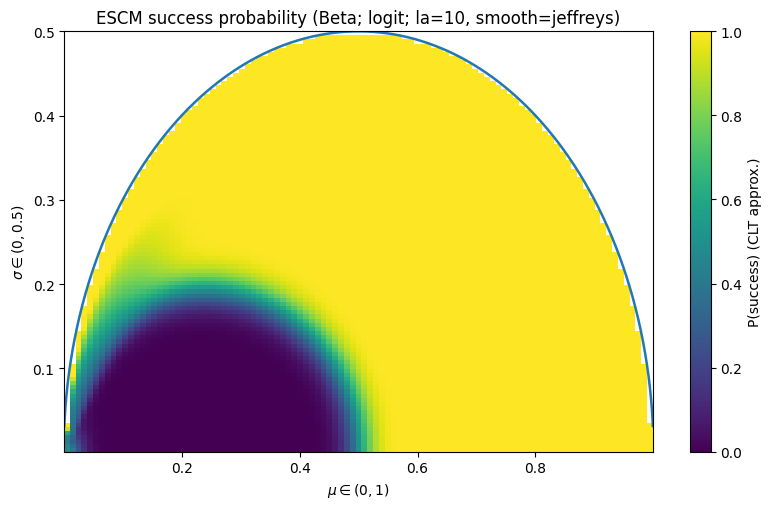}
  \caption{Approximate ESCM success probability}
\end{subfigure}\hfill
\begin{subfigure}[t]{0.48\linewidth}
  \centering
  \includegraphics[width=\linewidth]{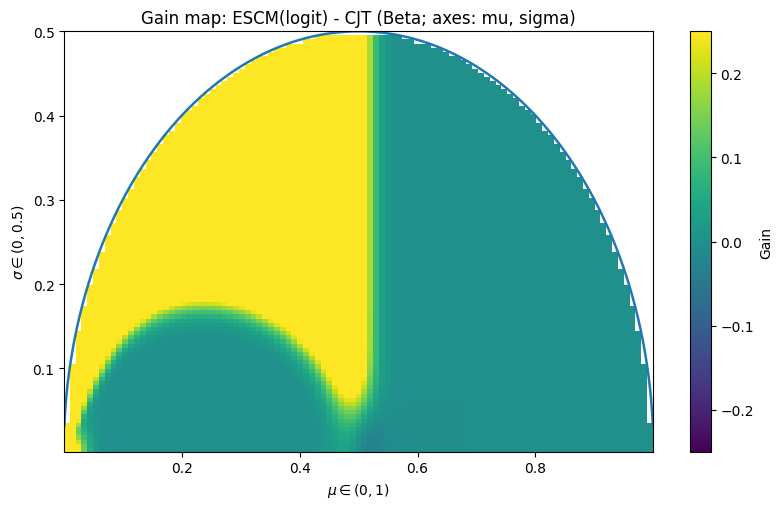}
  \caption{Gain relative to the equal-weight benchmark}
\end{subfigure}
\caption{ESCM with regularized log-odds coefficients under Beta-distributed competence ($l_a=10$, $\varepsilon=0.01$).}
\label{fig:beta_logit}
\end{figure}

%=============================
%=============================
\section{Segmented Competence and Mixture Models}
\label{sec:cmm}

\subsection{The Competence Mixture Model}
\label{subsec:cmm}

While the Beta distribution provides a natural benchmark for unimodal
heterogeneity, many electorates exhibit segmented informational structures, in
which competence is clustered into distinct groups. Such segmentation may arise
from differences in education, media environments, or domain-specific
expertise, and is frequently invoked in discussions of informed minorities and
issue-specific asymmetries.%\footnote{These examples are intended as stylized
%motivations rather than as empirical generalizations about all electorates.}
To capture this form of heterogeneity in a controlled way, we adopt a
competence mixture model (CMM) as a structured benchmark environment.

Let $p_i \in (0,1)$ and assume
\[
p_i \sim \sum_{c=1}^{C} \pi_c\, f_c(p),
\qquad
\sum_{c=1}^{C} \pi_c = 1,
\]
where $f_c$ denotes the competence distribution of group $c$ and $\pi_c$ its
population share. We focus on the case $C=3$, corresponding to a
low-competence group, an intermediate group, and a highly competent group.
Mixture specifications of this kind arise naturally when electorates contain
subpopulations that differ systematically in informational quality
\citep{nitzan1982,bolandMajority1989,estlundEpistemic2018,bovens2003,dietrichDeliberation2025}.

In the reduced-form framework of Section~\ref{sec:clt-general}, the
equal-weight benchmark depends only on the average competence level
\[
\bar{\mu} = \int_0^1 p\,f(p)\,dp.
\]
Segmented competence mixtures therefore provide a demanding benchmark for ESCM:
if the Gaussian approximation predicts any improvement over the equal-weight
rule, that improvement must come from exploiting information embedded in the
distribution of competence beyond its mean.

\subsection{The CMM-3 Benchmark}
\label{subsec:cmm-method}

We consider equal population shares
\[
\pi_1=\pi_2=\pi_3=\frac13,
\]
with the intermediate group fixed at $\mu_2=0.5$ and the remaining groups
centered at $\mu_1<0.5$ and $\mu_3>0.5$. Each component is modeled as a
truncated Gaussian distribution on $(0,1)$ with standard deviation
\[
\sigma_{\mathrm{comp}} = 0.12.
\]
These relatively wide component distributions generate substantial overlap
across groups, thereby creating a demanding environment for
competence-sensitive aggregation: group boundaries are not sharp signal
separators, and any gain from ESCM must be achieved despite substantial
within-group and cross-group noise.

The analysis scans the $(\mu_1,\mu_3)$ plane with
\[
\mu_1 \in (0,0.5),
\qquad
\mu_3 \in (0.5,1),
\]
holding $\mu_2$ fixed. Horizontal movement increases the competence of the
least informed group, vertical movement increases that of the most informed
group, and diagonal movement increases separation while leaving the
intermediate group unchanged.

Within the Gaussian approximation of Section~\ref{sec:clt-general}, the
equal-weight benchmark depends only on $\bar{\mu}$, whereas ESCM can in
principle respond to how competence is distributed across groups through
endogenous reweighting. CMM-3 thus isolates a setting in which structured
segmentation is present but the equal-weight rule responds only to the overall
mean competence.

\subsection{Results under CMM-3 Competence}
\label{subsec:cmm-results}

All results fix $n=501$. For the ESCM specifications considered below, the
assessment length is set to $l_a=10$, and this parameter enters through the
assessment-induced mapping from competence to aggregation coefficients rather
than through an explicit simulation of the full institutional procedure. For
each specification of the competence distribution and coefficient map, the
reduced-form moments $\mu_S$ and $\sigma_S^2$ of
Section~\ref{sec:clt-general} are computed by numerical integration over $f$.
For the regularized log-odds specification, the regularization parameter is set
to $\varepsilon=0.01$. The approximate success probability is then evaluated as
\[
\Pr(S_n>0)\approx \Phi\!\left(\sqrt{n}\,\frac{\mu_S}{\sigma_S}\right),
\]
so that the figures below should be interpreted as diagnostics for the
large-sample Gaussian approximation rather than as finite-sample guarantees for
an implemented ESCM.

\paragraph{Equal-weight benchmark.}
Figure~\ref{fig:cmm_cjt} shows that the approximate success probability under
the equal-weight rule depends only on $\bar{\mu}$ and is insensitive to
competence segmentation. Large regions of the $(\mu_1,\mu_3)$ plane therefore
yield only modest collective accuracy even when a highly competent group is
present, provided average competence remains near the indifference threshold.
This illustrates the demanding nature of the mixture benchmark: within the
Gaussian approximation, the equal-weight rule responds only to average
competence, leaving room for improvement only if a weighting scheme can exploit
how competence is distributed beyond its mean.

\begin{figure}[t]
\centering
\includegraphics[width=0.48\linewidth]{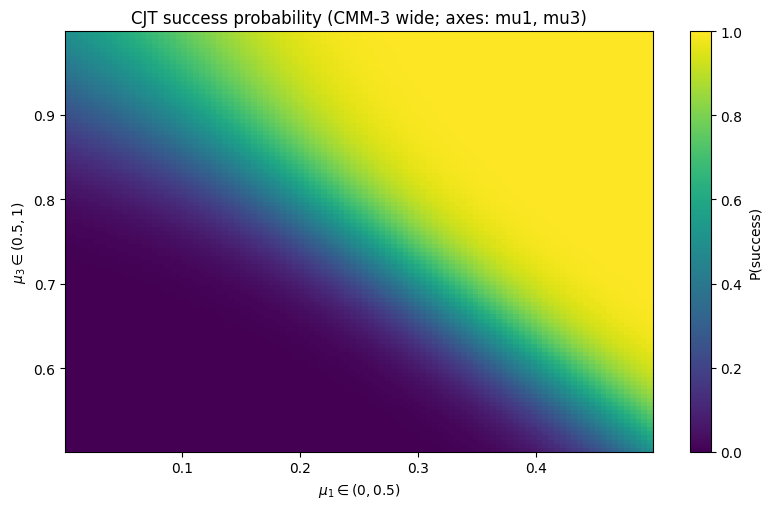}
\caption{Approximate equal-weight success probability under CMM-3 competence.}
\label{fig:cmm_cjt}
\end{figure}

\paragraph{ESCM with linear coefficients.}
Figure~\ref{fig:cmm_linear} reports ESCM performance under the linear
specification
\[
g(\bar{s})=\bar{s}.
\]
Unlike the equal-weight benchmark, the induced approximate success probability
depends non-trivially on $(\mu_1,\mu_3)$: performance improves as $\mu_3$ rises
even when $\mu_1$ remains low, indicating that, within this structured
benchmark, ESCM exploits the presence of a more informed subgroup. Gains are
concentrated where the equal-weight rule is most fragile and diminish where
average competence already implies high approximate accuracy.

\begin{figure}[t]
\centering
\begin{subfigure}[t]{0.48\linewidth}
  \centering
  \includegraphics[width=\linewidth]{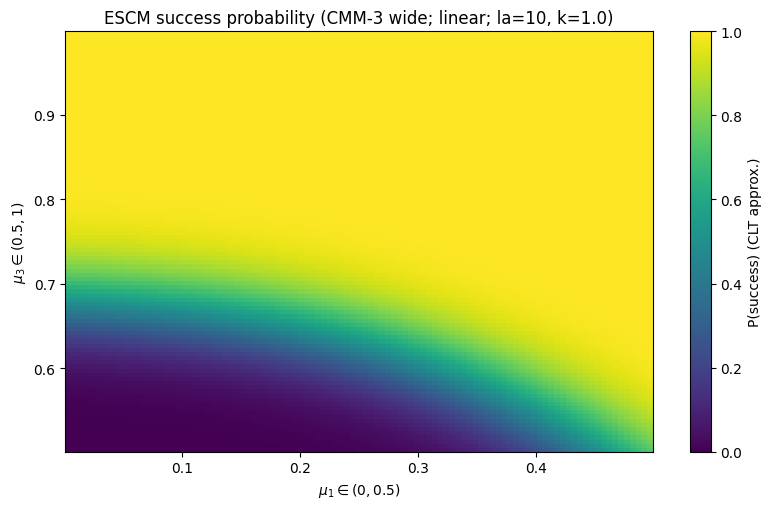}
  \caption{Approximate ESCM success probability}
\end{subfigure}\hfill
\begin{subfigure}[t]{0.48\linewidth}
  \centering
  \includegraphics[width=\linewidth]{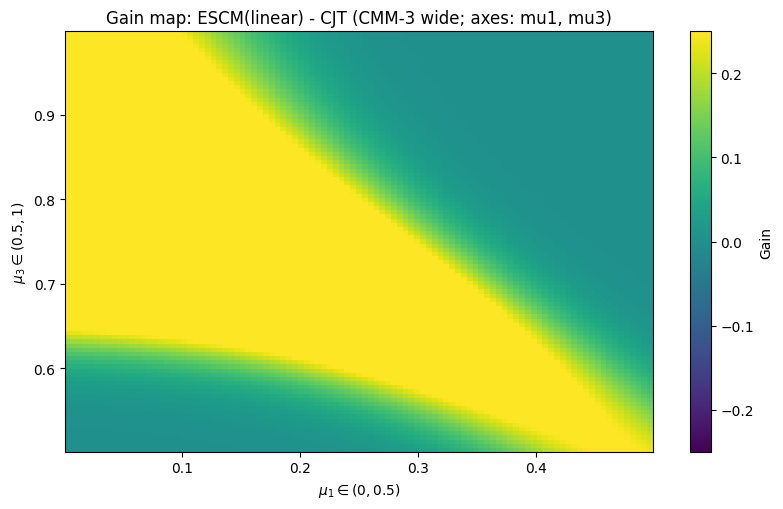}
  \caption{Gain relative to the equal-weight benchmark}
\end{subfigure}
\caption{ESCM with linear coefficients under CMM-3 competence ($l_a=10$, $k=1$).}
\label{fig:cmm_linear}
\end{figure}

\paragraph{ESCM with regularized log-odds coefficients.}
Figure~\ref{fig:cmm_logit} reports ESCM performance under the regularized
log-odds specification. Relative to the linear specification, gains are larger
and extend across a wider portion of the $(\mu_1,\mu_3)$ plane. This pattern is
consistent with a more selective competence-sensitive transformation of
assessment scores in this numerical environment. Since the regularized log-odds
map may generate negative coefficients for low scores, the mechanism is best
interpreted here as a form of signed evidence aggregation rather than as a
non-negative weighted-majority rule. The regularization parameter
$\varepsilon=0.01$ keeps the induced coefficients bounded near the score
boundaries.

\begin{figure}[t]
\centering
\begin{subfigure}[t]{0.48\linewidth}
  \centering
  \includegraphics[width=\linewidth]{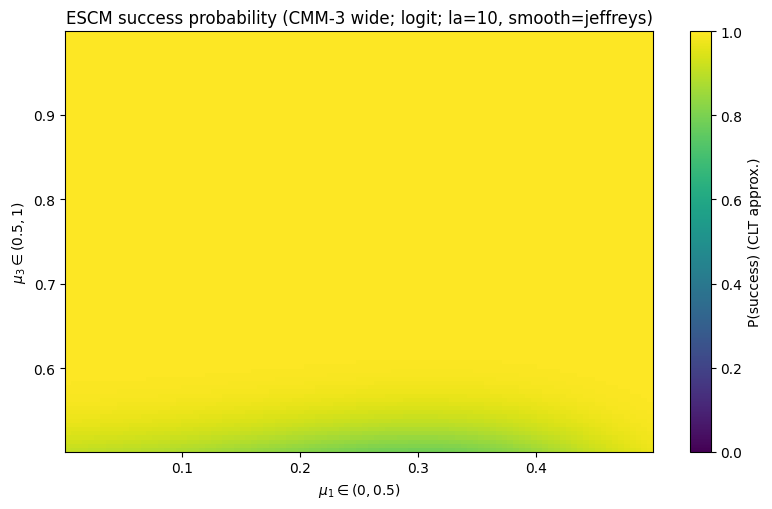}
  \caption{Approximate ESCM success probability}
\end{subfigure}\hfill
\begin{subfigure}[t]{0.48\linewidth}
  \centering
  \includegraphics[width=\linewidth]{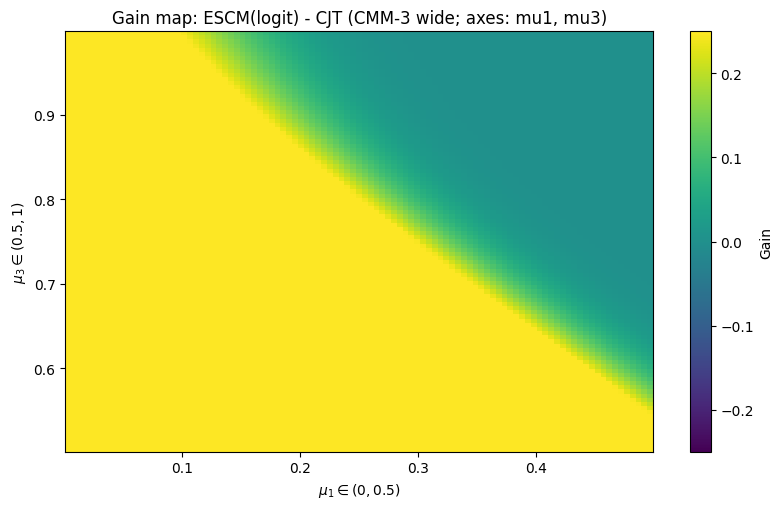}
  \caption{Gain relative to the equal-weight benchmark}
\end{subfigure}
\caption{ESCM with regularized log-odds coefficients under CMM-3 competence ($l_a=10$, $\varepsilon=0.01$).}
\label{fig:cmm_logit}
\end{figure}

%=============================================

%=============================================
\section{A General Mean-Improvement Result}
\label{sec:dominance}

Sections~\ref{sec:unimodal} and~\ref{sec:cmm} provide numerical evidence that,
for the specific coefficient maps examined in this paper, competence-sensitive
reweighting can improve the Gaussian large-sample approximation to collective
accuracy under both unimodal and segmented competence distributions. The
purpose of the present section is narrower: it isolates a weaker but fully
general analytical mechanism underlying those gains. Specifically, it shows
that, under heterogeneous competence, bounded monotone reweighting increases
the mean signed aggregate signal relative to the equal-weight benchmark.

Throughout this section, coefficient maps are normalized so that
\[
\mathbb{E}_f[w(p)] = 1.
\]
Because positive scalar rescaling multiplies both the mean and the standard
deviation of the reduced-form aggregate signal by the same factor, this
normalization leaves the Gaussian criterion $\mu_S/\sigma_S$ unchanged while
making comparisons across rules transparent.

\begin{proposition}[Mean improvement under competence-sensitive weighting]
\label{prop:mean-improvement}
Under Assumption~\ref{ass:regularity} and the normalization
$\mathbb{E}_f[w(p)] = 1$, let
\[
\mu_S(w) := \mathbb{E}_f[w(p)(2p-1)]
\]
denote the mean signed aggregate signal induced by the coefficient map $w$, and
let $w^{\mathrm{EW}} \equiv 1$ denote the equal-weight benchmark. If
$w^{\mathrm{ESCM}}(p)$ is bounded, non-decreasing, and non-constant on the
support of $f$, then:
\begin{enumerate}[label=(\roman*)]
\item \textbf{Weak mean dominance.}
\[
\mu_S(w^{\mathrm{ESCM}}) \ge \mu_S(w^{\mathrm{EW}}).
\]

\item \textbf{Strict improvement under heterogeneity.}
If $\mathrm{Var}_f(p)>0$, then
\[
\mu_S(w^{\mathrm{ESCM}}) > \mu_S(w^{\mathrm{EW}}).
\]

\item \textbf{Degenerate case.}
If $p_i = p_0$ for all $i$, then
\[
\mu_S(w^{\mathrm{ESCM}})=\mu_S(w^{\mathrm{EW}}),
\]
so competence-sensitive weighting yields no mean advantage.
\end{enumerate}
\end{proposition}

\begin{proof}
Write
\[
h(p):=2p-1,
\]
which is strictly increasing on $[0,1]$. Then
\[
\mu_S(w)=\mathbb{E}_f[w(p)h(p)].
\]
Under the normalization $\mathbb{E}_f[w(p)]=1$, we may write
\[
\mathbb{E}_f[w(p)h(p)]
=
\mathbb{E}_f[h(p)]
+
\mathrm{Cov}_f(w(p),h(p)).
\]
Since $w^{\mathrm{ESCM}}$ is non-decreasing and non-constant on the support of
$f$, while $h$ is strictly increasing, Chebyshev's association inequality for
comonotone functions implies
\[
\mathrm{Cov}_f\!\bigl(w^{\mathrm{ESCM}}(p),h(p)\bigr)\ge 0,
\]
with strict inequality whenever $\mathrm{Var}_f(p)>0$. Hence
\[
\mu_S(w^{\mathrm{ESCM}})
=
\mu_S(w^{\mathrm{EW}})
+
\mathrm{Cov}_f\!\bigl(w^{\mathrm{ESCM}}(p),h(p)\bigr),
\]
which proves parts (i) and (ii). Part (iii) is immediate: if
$f=\delta_{p_0}$, then both $w^{\mathrm{ESCM}}(p)$ and $h(p)$ are constant
$f$-almost surely, so the covariance term vanishes.
\end{proof}

\begin{corollary}[Sufficient condition for Gaussian-accuracy improvement]
\label{cor:mean-to-snr}
Under the assumptions of Proposition~\ref{prop:mean-improvement}, if
\[
\mu_S(w^{\mathrm{ESCM}})\,\sigma_S(w^{\mathrm{EW}})
>
\mu_S(w^{\mathrm{EW}})\,\sigma_S(w^{\mathrm{ESCM}}),
\]
then
\[
\frac{\mu_S(w^{\mathrm{ESCM}})}{\sigma_S(w^{\mathrm{ESCM}})}
>
\frac{\mu_S(w^{\mathrm{EW}})}{\sigma_S(w^{\mathrm{EW}})},
\]
and therefore the Gaussian approximation of
Corollary~\ref{cor:robustness} assigns a higher large-sample success
probability to ESCM than to the equal-weight benchmark.
\end{corollary}

\begin{proof}
Immediate from rearranging the inequality and applying
Corollary~\ref{cor:robustness}.
\end{proof}

The condition in Corollary~\ref{cor:mean-to-snr} has a simple interpretation:
endogenous weighting improves the CLT-based Gaussian criterion whenever the
proportional gain in the mean signed signal exceeds the proportional increase
in its dispersion. In other words, mean improvement translates into higher
approximate large-sample accuracy as long as the induced increase in dispersion
remains sufficiently limited.

Proposition~\ref{prop:mean-improvement} is therefore intentionally more modest
than a universal signal-to-noise-ratio dominance theorem, but it remains
substantive. It identifies a distribution-general mechanism that, under
heterogeneous competence and the stated assumptions, always raises the mean
signed aggregate signal relative to equal weighting: bounded monotone
reweighting creates a positive covariance between competence $p$ and signed
signal quality $2p-1$.

Sections~\ref{sec:unimodal} and~\ref{sec:cmm} then show numerically that, for
the specific coefficient maps and competence environments studied in this
paper, this mean improvement is often accompanied by only moderate increases in
dispersion and thus by higher approximate large-sample accuracy.

A related distinction concerns optimality. The Nitzan--Paroush benchmark is
likelihood-oriented, whereas maximizing $\mu_S(w)/\sigma_S(w)$ is a different
optimization problem. Proposition~\ref{prop:asymptotic} should therefore be
read as an asymptotic relation between ESCM score transforms and
likelihood-based reliability benchmarks, not as a result showing that log-odds
coefficients globally maximize the CLT-based Gaussian criterion.

%==============================
%==============================
\section{Conclusion}
\label{sec:conclusion}

This paper has proposed the Epistemic Shared-Choice Mechanism (ESCM) as a
transparent institutional architecture for implementing endogenous bounded
weighting in binary collective decisions. Rather than presupposing direct
knowledge of individual competence or fixing hierarchies of influence
\emph{ex ante}, ESCM uses assessment performance as an observable proxy for
latent reliability and maps it into issue-specific aggregation coefficients.

The main general analytical result, Proposition~\ref{prop:mean-improvement}, shows that, under heterogeneous
competence, any bounded non-decreasing coefficient map increases the mean
signed aggregate signal relative to the equal-weight benchmark. The numerical analyses show that, for the specific coefficient maps
examined here, this mean improvement is often associated with higher Gaussian
large-sample accuracy under both unimodal and segmented competence
distributions.

Using central limit approximations, the analysis has clarified how epistemic
performance depends on the informational structure of the population. In
unimodal competence environments, endogenous reweighting produces relatively
modest gains over the equal-weight rule. In segmented competence environments,
captured here through competence mixture models, the gains can be substantially
larger. In such settings, competence-sensitive aggregation allows
better-informed minorities to exert greater influence on the aggregate signal
even when average competence remains close to randomness.

What distinguishes ESCM from prior weighted aggregation proposals is the
combination of three features it is designed to satisfy simultaneously:
aggregation coefficients are generated endogenously from observable assessment
performance, they are bounded by design, and they are issue-specific rather
than fixed across decision domains. In this sense, the paper identifies a
general mechanism through which competence-sensitive reweighting can improve
collective performance: by assigning weakly greater aggregation coefficients to
more competent participants, it raises the mean signed aggregate signal
whenever competence is heterogeneous.

At the same time, the analysis also clarifies the paper's limits. The formal
results are derived in a stylized reduced-form representation of endogenous
weighting, and the assessment procedure is treated as a procedural
implementation rather than as a fully validated psychometric institution. The
contribution is therefore best understood as establishing a general
mean-improvement result together with numerical evidence that, in structured
competence environments, this mechanism can translate into higher approximate
large-sample accuracy.

\subsection*{Directions for Future Research}

The present analysis deliberately abstracts from several dimensions that warrant separate investigation.

First, the psychometric validity of the assessment stage deserves closer study.
Key issues include the reliability of peer-generated item pools as estimators
of latent competence, the sensitivity of the resulting coefficient map to
assessment noise, and the behavior of alternative scoring rules under more
realistic response models.

Second, the strategic robustness of ESCM remains to be characterized. The
present analysis abstracts from incentives to manipulate item authorship,
peer-review evaluations, or assessment responses. Formal results on
incentive-compatibility and on the coalition size required to overturn a
correct collective outcome would substantially strengthen the mechanism's
institutional foundations.

Third, the analysis assumes conditional independence of individual signals.
Extending ESCM to settings with informational cascades
\citep{banerjeeSimple1992,bikhchandaniTheory1992}, media-driven correlation
\citep{gagrcinDefending2024}, or network-mediated dependence
\citep{acemogluOpinion2011} would clarify when endogenous reweighting continues
to add epistemic value and when it instead amplifies redundancy or common
informational shocks.

Fourth, while the numerical analyses provide evidence within the Gaussian
approximation, a systematic finite-sample simulation study varying population
size, assessment length, and noise levels would more precisely characterize
when the CLT benchmark is reliable and how closely it tracks realized
finite-sample performance \citep{grimEpistemic2024}.

Fifth, the present paper focuses on binary collective decisions. Extending the
framework to multi-option settings would require moving from a scalar
competence model to richer confusion structures and from binary aggregation to
multi-class competence-sensitive decision rules.

%\vspace{2mm}
%\textit{Acknowledgments:} The author would like to thank anonymous referees of previous submissions and the community discussion for improving this theory with insights, reflections and challenging questions.

\backmatter

%% Declarations section (fill as needed for the journal)
\section*{Declarations}
\begin{itemize}
\item \textbf{Funding} Not applicable.
\item \textbf{Conflict of interest} The author declares no competing interests.
\item \textbf{Ethics approval and consent to participate} Not applicable.
\item \textbf{Consent for publication} Not applicable.
\item \textbf{Data availability} Not applicable.
\item \textbf{Materials availability} Not applicable.
\item \textbf{Code availability} Available upon request.
\item \textbf{Author contribution} Sole author.
\end{itemize}

%% Bibliography
%% The sn-jnl class controls the bibliography style; do NOT set \bibliographystyle here.
\bibliography{ESCM_Biblio}

%% Optional author blurb (can be moved to a footnote or Acknowledgments)
\vspace{1em}
{\small ENRICO MANFREDI, Ph.D. in Mathematics, University of Bologna. E-mail: \href{mailto:enrico.manfredi@gmail.com}{enrico.manfredi@gmail.com}.}

\end{document}